\documentclass[12pt]{article}
\usepackage{pudlatex}
\usepackage{fullpage}
\usepackage{amsmath}

\newcommand\dist{\mathop{\rm dist}\nolimits}
\newcommand\wt{\mathop{\rm wt}\nolimits}
\newcommand\Ent{\mathop{\rm H}\nolimits}

\title{Linear tree codes and the problem of explicit constructions}

\author{Pavel Pudl\'ak
\thanks{The author is supported by
the grants P202/12/G061 of GA \v{C}R and
RVO: 67985840.}}

\begin{document}
\maketitle

\begin{abstract}
We reduce the problem of constructing asymptotically good tree codes to
the construction of triangular totally nonsingular matrices over fields
with polynomially many elements. We show a connection of this problem
to Birkhoff interpolation in finite fields.
\end{abstract}


\section{Introduction}

Tree codes, in the sense we are going to use in this paper, were
introduced by L.J. Schulman in 1993.  He showed that asymptotically
good tree codes can be used in efficient interactive communication
protocols and proved by a probabilistic argument that such tree codes
exist \cite{schulman93,schulman}. He posed as an open problem to give
an explicit effectively computable construction of them. Efficiently
constructible tree codes would be very useful in designing robust
interactive protocols. The field has attracted a lot of attention in
recent years, however this central problem still remains open. A
possible solution may be the construction of Moore and
Schulman~\cite{moore-schulman} found recently. Their construction
provides asymptotically good tree codes if a certain
number-theoretical conjecture, introduced in their paper, is true. The
conjecture is inspired by some well-known results about exponential
sums and is supported by numerical evidence.

In this paper we propose a different approach to this problem. We
study generator and parity check matrices of linear codes and reduce
the problem to constructing triangular totally nonsingular matrices
over fields of polynomial size. A lower triangular matrix $M$ is
called triangular totally nonsingular if every square submatrix of $M$
whose diagonal is entirely in the lower triangle is
nonsingular. Explicit examples of such matrices are known over the
field of real numbers, and these include matrices with integral
elements. One can also show that triangular totally nonsingular
matrices exist over finite fields of exponential size. The question
whether they exist over finite fields of polynomial size (or at least
subexponential size) is open. Since totally nonsingular matrices
(i.e., matrices whose \emph{all} square submatrices are nonsingular) do exist
over fields of linear size, we conjecture that there exist triangular
totally nonsingular matrices over fields of polynomial size.

In this way we may be reducing the problem of constructing tree codes
to a more difficult problem. But since the concept of triangular
totally nonsingular matrices is very natural, the problem of
constructing such matrices over small fields is of independent
interest. We also hope that due to this connection we will be able
draw attention of the linear algebra community to this important open
problem in coding theory.

Here is a brief overview of the paper. In Section~1 we define linear codes
and prove some basic facts about them. Some facts in this section are
well-known, or well-known in some form. In particular, the existence
of asymptotically good linear tree codes was first proved by Schulman.
In Section~2 we observe that one can concatenate a tree code with a
constant size alphabet and input length $\log n$ with a tree code with
an alphabet of polynomial size and input length $n$ in order to obtain
a tree code with a constant size alphabet and input length $O(n\log
n)$. Since the ``short'' tree code can be found by brute force search
in polynomial time, it suffices to construct in polynomial time an
asymptotically good tree code with an alphabet of polynomial size in
order to get a polynomial time construction of asymptotically good
tree codes. This is also a well-known fact and is included for the sake of
completeness. 

In the main part of the paper we focus on linear tree codes of rate
$1/2$. In Section~4 we give a characterization of parity check
matrices of linear tree codes with a given minimum distance. In
Section~5 we introduce MDS linear tree codes.  We show that an MDS
linear tree code of rate $1/2$ is determined by a triangular totally
nonsingular matrix.  Since the minimum distance
of rate $1/2$ MDS tree codes is greater than $1/2$, in order to solve
the construction problem, it suffices to construct triangular totally
nonsingular matrices over fields of polynomial size.  We discuss some
approaches to the problem of constructing such matrices in Section~6.
In the last section we show a connection between MDS linear tree codes
and the Birkhoff interpolation problem.

\subsection*{Acknowledgment}
The author would like to thank Miroslav Fiedler, Leonard Schulman, Madhu
Sudan and an anonymous referee for their remarks and suggestions.

\section{Basic concepts and facts}\label{sect2}

We will assume that the reader is familiar with the basic concepts and
results from the theory of block codes. (The reader can find missing
definitions, e.g., in~\cite{macwilliams-sloane}.)

A \emph{tree code} of input length $n$ with finite alphabets $\Pi$ and
$\Sigma$ is a mapping $c:\Pi^n\to\Sigma^n$ of the form
\[
c(x_1\dts x_n) = (c_1(x_1),c_2(x_1,x_2),\dots,c_n(x_1\dts x_n))
\]
where $c_i:\Pi^i\to\Sigma$ and 
\bel{e-2.1}
(x_1\dts x_i)\ \mapsto\ (c_1(x_1)\dts c_i(x_1\dts x_i))
\ee
is a one-to-one mapping for every $i=1\dts n$. Hence $c$ induces an
isomorphism of the tree of the input words onto the tree of output
words, the code words of $c$. 

A natural way to define tree codes is to
define them as mappings of infinite sequences to infinite sequences,
i.e., $c:\Pi^\omega\to\Sigma^\omega$. This would somewhat complicate
the relations to the concepts in linear algebra that we want to use,
so we prefer the definition with finite strings, although most of the
concepts and results presented here can easily be translated to the
infinite setting.

%The following is an important concept introduced in~\cite{schulman}. 
Let $c$ be a tree code of input length $n$. Let $C$ be the set of the 
code words, i.e., the range of the function $c$. Then 
the \emph{minimum relative distance} of the tree code $c$, denoted by
$\delta(c)$,  is the
minimum over all $0\leq k< l\leq n$, $u\in\Sigma^k$,
$v,v'\in\Sigma^{l-k}$, $w,w'\in\Sigma^{n-l}$,
$(u,v,w),(u,v',w')\in C$, $v_1\neq v'_1$ of the quantity
\[
\frac{\dist(v,v')}{l-k},
\]
where $\dist(x,y)$ denotes the Hamming distance and $v_1$ and $v'_1$
are the first elements of the strings $v$ and $v'$. 

The \emph{rate} of the tree code is
\[
\rho(c)=\frac{\log|\Pi|}{\log|\Sigma|}.
\]

\begin{definition}
A tree code $c$ is \emph{linear,} if $\Pi$ and $\Sigma$ are finitely
dimensional vector spaces over a finite field $F$ and $c$ is a linear
mapping.  
\end{definition}

It should be noted that convolutional codes are special instances of
linear tree codes, but they are not interesting for us, because their
minimum relative distance, as defined above, is very small.

In this paper we will focus on the codes where $\Pi$ is the field $F$
and $\Sigma=F^d$. In this case, the rate of a linear code is the inverse
of the dimension of $\Sigma$, i.e., $\rho(c)=1/d$.

As in linear block codes, the minimum relative distance is
characterized by the minimum weight of nonzero code words: the minimum
relative distance of a linear tree code $c$ is the
minimum over all $0\leq k< l\leq n$, 
$v\in\Sigma^{l-k}$, $w\in\Sigma^{n-l}$,
$(\bar{0}^k,v,w)\in C$, $v_1\neq \bar{0}$ of 
\[
\frac{\wt_\Sigma(v)}{l-k},
\]
where $\bar{0}^k$ is the zero vector in $\Sigma^k$ and $\wt_\Sigma$ denotes
the Hamming weight with respect to the alphabet $\Sigma$. Note that it
is also natural to consider the Hamming weight with respect to
$F$. So we define  $\tilde{\delta}(c)$ as the minimum of
\[
\frac{\wt_F(v)}{d(l-k)}
\]
and focus on this quantity in the rest of this paper.
Clearly $\tilde{\delta}(c)\leq\delta(c)$. 

\begin{theorem}\label{t1}
  Let $n\geq 1$, $q=|F|$, $r=q^d=|\Sigma|$ and $0<\delta<\frac{r-1}r$
  such that \bel{e1} \log_r2q+\Ent_r(\delta)\leq 1.  \ee Then there
  exists a linear code $c:F^n\to\Sigma^n$ with $\tilde{\delta}(c) >
  \delta$. Moreover, if $q,r$ and $\delta$ are fixed, then such codes
  can be constructed for every $n$ in time $2^{O(n)}$.
\end{theorem}

{\it Remarks}

1. In the theorem, $\Ent_r$ denotes the $r$-entropy function defined by
\[
\Ent_r(x)=x\log_r(r-1)-x\log_rx-(1-x)\log_r(1-x).
\]

2. Peczarski \cite{peczarski} proved that for every prime power $q$,
there exist codes with relative distance $1/2$ and rate
$1/(2+\lceil\log_q 4\rceil)$. This is better than the bound in the
theorem above for $\delta=1/2$.

3. Note that there exists $\delta>0$ such that for every $q>2$, there
exist tree codes with rate $1/2$ (i.e., $d=2$) and minimum relative
distance $\geq\delta$. We do not know if binary (i.e., $q=2$) tree codes
with rate $1/2$ can have asymptotically positive minimum
relative distance.  

\begin{proof} (Our proof is different from the one presented
  in~\cite{schulman}, but the basic idea is essentially the same.)  
%Let $\Sigma=F^d$; thus $r=q^d$. 
  Suppose $q$, $r$ and $\delta$ satisfy the inequality \eqrf{e1}
  above.  We will prove the existence by induction on $n$. For $n=1$,
  take the repetition code (i.e., $c(x)=(x\dts x)\in\Sigma$). Now
  suppose we have such a code $c$ for $n$ and want to construct a code
  $c':F^{n+1}\to\Sigma^{n+1}$. We take $v\in\Sigma^n$ random and put
\[
c'(x_0,x_1\dts x_n):=(x_0\bar{1}^d,c(x_1\dts x_n)+x_0v).
\]
Here we denote by $\bar{1}^d$ the vector in $F^d$ whose all
the $d$ coordinates are $1$.

The minimum weight condition is satisfied in the case when $x_0=0$ by
the induction assumption. Thus we only need to satisfy the condition
for $x_0\neq 0$ by a suitable choice of $v$. As above, let $C$ denote
the range of $c$. Let $C|_{[dk]}$ denote the projection of the vectors
of $C$ on the first $dk$ coordinates.
\begin{lemma}
  Let $1\leq k\leq n$. Let $v\in\Sigma^k=F^{dk}$ be a uniformly
  randomly chosen vector. Let $\langle C|_{[dk]}\cup\{v\}\rangle$ be
  the span of $C|_{[dk]}$ and the vector $v$. Then the probability
  that $\langle C|_{[dk]}\cup\{v\}\rangle$ contains a nonzero vector
  $u$ with weight $\wt_{F}(u)\leq\delta dk$ is at most $2^{-k}$.
\end{lemma}
Note that we are counting nonzero coordinates with respect to the
field $F$, not $\Sigma$.
\begin{proof} 
  Since by the induction assumption, there is no $u$ in $C|_{[dk]}$
  whose weight is $\leq\delta dk$, such a vector must be a linear
  combination $aw+bv$ where $w\in C|_{[dk]}$ and $b\neq 0$. So
  $\dist(v,-ab^{-1}w)\leq\delta dk$. The number of vectors whose
  distance from $C|_{[dk]}$ is $\leq\delta dk$ is estimated by $q^k$,
  the cardinality of $C|_{[dk]}$, times the size of a ball of radius
  $\delta dk$ in $F^{dk}$, which we can bound using the entropy
  function by $r^{\Ent_r(\delta)k}$. Thus the probability is at most
  $2^{-k}$, if
\[
q^kr^{\Ent_r(\delta)k}/r^k\leq 2^{-k}, 
\]
which is equivalent to \eqrf{e1}.
\end{proof}
Now we can finish the proof of the existence of the tree code. The
probability that $\tilde{\delta}(c')\leq\delta$ is at most the
probability that, for some $k$, $\langle C|_{[dk]}\cup\{v\}\rangle$
contains a nonzero vector $u$ with weight $\wt_{F}(u)\leq\delta dk$,
which is, according to the lemma, at most $\sum_{k=1}^n2^{-k}<1$.

\bigskip
We now estimate the number of operations that are needed to find such
a code. For every $k=1\dts n-1$, we have to search $r^k$ vectors and we
have to determine their distances from $q^k$ vectors of the code from
the previous round. Thus we have to consider $\sum_{k=1}^{n-1}r^{k}q^k<r^{n}q^n$
cases, each of which takes polynomial time. Thus the time is $2^{O(n)}$.
\end{proof}

The generator matrix of a tree code $c$ is defined in the same way as
for ordinary codes. Let $e_i^n$ denote the vectors of the standard
basis of $F^n$, i.e., vectors $$(1,0\dts 0),(0,1\dts 0)\dts(0,0\dts
1).$$  The \emph{generator matrix} of $c$ is the $n\times dn$ matrix
whose rows are vectors $c(e_i^n)$. It is a block upper triangular matrix
where the blocks are $1\times d$ submatrices and the blocks on the
main diagonal are nonzero vectors of $F^d$, because the mappings
\eqref{e-2.1} are one-to-one.

We define \emph{cyclic tree codes} as the linear tree codes
that satisfy
\[
v\in C\ \Rightarrow\ \bar{0}v|_{[d(n-1)]}\in C.
\]
This means that with every code word $v$, the code contains a word
that is obtained by adding $d$ zeros at the beginning of $v$ and
deleting the last $d$ coordinates. (In this particular case it would
be better to use infinite sequences instead of finite ones.)
We observe that if $c$ is a
cyclic tree code with the space $C$ of the code words, there
exists a cyclic tree code $c'$ with the same code words whose
generator matrix is \emph{block-Toeplitz}. Indeed, define the
generator matrix of $c'$ as the shifts of $c(e_1)$. Formally, put
\[
c'(e_i^n):=(\bar{0}\dts\bar{0},c(e_1^n)|_{[n-i+1]})
\]
for $i=1\dts n$.

Note that \emph{convolutional codes} (see, e.g., \cite{vanLint}) are,
essentially, a special case of cyclic linear tree codes. To this end
we must consider linear tree codes of the form
$c:\Pi^\omega\to\Sigma^\omega$. Then $c$ is a convolutional code if it
is generated by a vector $c(e_1^\omega)$ that has only a finite number
of nonzero entries. Obviously, such a code cannot be asymptotically
good. 

\medskip We will now give a slightly different proof of the existence
of good linear tree codes with the additional property of
cyclicity. Note that in this proof we need only a linear number of
random bits.

\begin{proof}
  Let $v_2\dts v_n\in\Sigma$ be chosen uniformly randomly and
  independent. Thus $(v_2\dts v_n)$ is a random vector from
  $\Sigma^{n-1}$. Let $T:=T(\bar{1}^d,v_2\dts v_n)$ be the upper block
  triangular Toeplitz matrix with the first row equal to
  $\bar{1}^d,v_2\dts v_n$. Since $T$ is Toeplitz, we only need to ensure
  the condition about the number of nonzero elements in nonzero vectors
  for vectors with the first block nonzero. Let $1\leq k\leq n$ and
  let $T(\bar{1}^d,v_2\dts v_k)$ be the submatrix of $T$ determined by
  the first $k$ rows and the first $dk$ columns. We will estimate the
  probability that for a linear combination of the rows in which the
  first row has nonzero coefficient is a vector $u$ with $\leq\delta
  dk$ nonzero coordinates. The vector $u$ can be expressed, using
  matrix multiplication, as \bel{e3} u=(a_1\dts a_k)T(\bar{1}^d,v_2\dts
  v_k), \ee where $a_1\neq 0$. The vector $u$ has the form
  $(a_1\bar{1}^d,u_2\dts u_k)$, $u_i\in\Sigma$. Let $a_1\neq 0,a_2\dts
  a_k$ be fixed and view $v_2\dts v_k$ as variables. Then \eqrf{e3}
  defines a linear mapping from $F^{d(k-1)}$ to itself. Due to the
  form of the matrix $T(\bar{1}^d,v_2\dts v_k)$ and the fact that
  $a_1\neq 0$, the mapping is onto, hence the vector $(u_2\dts u_k)$
  is uniformly distributed. Thus we can use the Chernoff bound, or the
  bound by the entropy function, to estimate the probability for a
  fixed linear combination. Then use the union bound to estimate the
  probability that such a linear combination exists. The rest is the same
  computation as in the first proof.
\end{proof}

Parity-check matrices for linear tree codes are defined in the same
way as for ordinary codes: their row vectors are the vectors of some
basis of the dual space to the space of the code words $C$. (Thus 
parity-check matrices uniquely determine $C$, but, in general, not the
function $c$.)  We now describe a \emph{normal form of the parity-check
  matrices of linear tree codes.}

\begin{proposition}[Normal Form]\label{p-2.3}
Every linear tree code  $c:F^n\to F^{dn}$ has a parity-check matrix of
the following form: 
\bi
\item lower block triangular matrix with blocks of dimensions
  $(d-1)\times d$ and with blocks on the main diagonal of full rank
  $d-1$, 
\ei 
Vice versa, any matrix satisfying the condition above is
  a parity-check matrix of a linear tree code $c:F^n\to F^{dn}$.
\end{proposition}
\begin{proof}
Let $M$ be a parity-check matrix of a tree code. We will transform $M$
into the form described above using row operations, i.e., we will use
Gaussian elimination to rows. 

The matrix $M$ has dimensions $(d-1)n\times dn$ because its rows span
a vector space dual to $C$ and $C$ has dimension $n$.  The basic
property of the matrix is: 
\bi
\item[(*)] for every $1\leq k\leq n$ the last $dk$ columns of $M$ span
  a vector space of dimension $(d-1)k$.  \ei To prove (*), consider
  the matrix $M'$ consisting of the last $dk$ columns. Let $L$ be the
  row space of $M'$. The dual space $L^\bot$ is the space of all
  vectors $v\in F^{dk}$ such that $(\bar{0}^{d(n-k)},v)\in C$, because
  $M$ is a parity-check matrix of $C$. Its dimension is at least $k$,
  because it contains the $k$ linearly independent projections of
  vectors $c(e_i^n)$, $i=(n-k)+1\dts n$. It also is at most $k$,
  because every linear combination of generating vectors that contains
  some $c(e_i^n)$, $i\leq n-k$, with a nonzero coefficient has a
  nonzero coordinate outside of the last $dk$ positions. Thus, indeed,
  the dimension of $L$ is $(d-1)k$.

% This is
% because the rows of this submatrix span a space dual to the subspace
% of the vectors of $C$ that have zeros on the first $d(n-k)$
% coordinates. That space has the dimension precisely $k$ because there
% are $k$ vectors with this property in the generator matrix and the
% remaining vectors in the generator matrix restricted to the first
% $(n-k)d$ coordinates are linearly independent.

\medskip
We start the elimination process with the last $d$ columns. Since the
rank of this matrix is $d-1$, we can eliminate all rows of this
$d\times (d-1)n$ matrix except for $d-1$ ones that form a basis of the
row space. We permute the rows so that these $d-1$ rows are at the
bottom. Now consider the submatrix $M'$ with the first $(d-1)(n-1)$ rows
and the first $d(n-1)$ columns of the transformed parity-check matrix
and the submatrix $N'$ of the generator matrix $N$ of $C$ with the
first $n-1$ rows and the first $d(n-1)$ columns. The matrix $M'$ has full
rank, because $M$ has it, and the rows of $M'$ are orthogonal to
the rows of $N'$. Hence $M'$ is a parity-check matrix of the code
defined by $N'$. So we can assume as the induction hypothesis that $M'$
can be transformed into a normal form. Thus $M$ has been transformed
into a normal form.

\medskip Now we prove the opposite direction. Let $M$ be a matrix
satisfying the condition of the proposition (in fact, we will be using
the property (*) that follows from it). We will construct a
generator matrix of a tree code $c$ starting from the last row of the
matrix and going upwards. Let $v\in\Sigma$ be a nonzero vector that is
orthogonal to the row space of the submatrix of $M$ consisting of the
last $d$ columns. We define $c(e_n)$ to be $v$ preceded with $d(n-1)$
zeros. Suppose we already have $c(e_{n-k+1})\dts c(e_n)$. We take any
vector $u\in F^{d(k+1)}$ that is orthogonal to the row space of the
submatrix of $M$ consisting of the last $d(k+1)$ columns and is
independent of the vectors $c(e_{n-k+1})\dts c(e_n)$ restricted to the
last $d(k+1)$ coordinates.  The vector $u$ must have some nonzero on
the first $d$ coordinates, because $c(e_{n-k+1})\dts c(e_n)$ span the
space dual to the row space of $M$ restricted to the last $dk$
columns.  Then we define $c(e_{n-k-1})$ to be $u$ preceded with
$d(n-k-1)$ zeros.
\end{proof}

% One can also define \emph{dual codes}. Given a code $c:F^n\to F^{2n}$, we
% take the parity-check matrix of $c$ of the form $1^\circ$, take the
% columns in the reverse order and use this matrix as a
% generator matrix of the dual code $c^\bot$. In order to define dual
% codes in general, one has to use codes $c:\Pi^n\to\Sigma^n$ where
% $\Pi$ and $\Sigma$ are general finitely dimensional vectors spaces over a
% finite field $F$; we leave this to the reader.

\section{From a large alphabet to a small one}

% Except for the construction of \cite{moore-schulman}, which is based
% on an unproven conjecture, no other explicit construction of an infinite
% family of tree codes over a fixed alphabet $\Sigma$ which have minimum
% relative distance $\delta$ for some constant $\delta>0$ is known. ``Explicit''
% is a vague concept; one interpretation is a polynomial time
% construction of the generator matrix. We will show a reduction of the
% latter problem to a construction of an infinite family of linear
% tree codes with positive rate and relative distance over fields of
% polynomial size.  We will start with linear codes.

It is well-known that it suffices to construct an asymptotically good
tree code whose input and output alphabets have polynomial sizes in
order to construct an asymptotically good tree code with finite
alphabets. The resulting construction is not quite explicit, because
it relies on the construction of small tree codes by brute-force
search, but it can produce the code in polynomial time. We present
this reduction for the sake of completeness and also in order to check
that it works for linear codes. For
simplicity, we will restrict ourselves to the binary input
alphabet and finite fields of characteristic 2.

\begin{proposition}
  Let $b,d$ and $\delta>0$ be constants, then there exist constants
  $d'$ and $\delta'>0$ such that the following is true. Suppose a
  generator (or parity-check) matrix of a linear tree code
  $c:\F_{2^\ell}^n\to \F_{2^\ell}^{dn}$ is given, where $\ell\leq
  b\log n$ and the minimum relative distance of $c$ is~$\delta$.  Then
  one can construct in polynomial time a generator matrix of a binary
  linear tree code $c':\F_2^{n'}\to\F_2^{d'n'}$ where $n'=\ell (n+1)$
  and the minimum relative distance of $c'$ is $\delta'$.
\end{proposition}
\begin{proof}
The basic idea is very simple: replace the symbols of the long
code $c$ by bit-strings of a binary code $a$ of logarithmic length. The short
code $a$ can be found in polynomial time in $n$ by Theorem~\ref{t1}
because its length is logarithmic. However, in order to make this idea
work, one has to overcome some technical problems. Remember that we
have to ensure large weight on all intervals. If we simply replaced
the symbols by code-words of some code, we would not be able to ensure large
weight on intervals that are parts of two consecutive
strings corresponding to two consecutive symbols of the code $c$. 

Our solution is to use two mechanisms to ensure large weights---one for
short intervals and one for long intervals. We reserve odd bits for
the first mechanism and even bits for the second. For short intervals
it is convenient to use a cyclic tree code, because we can take any shift
of a fixed string. In other words, we stretch a short cyclic tree code $a$
over the full length $n'$. In this way we ensure large weights on
short intervals. Then we do not need to encode the symbols of the long
code by a words of a tree code; it suffices to use some block code $f$.

Here is the construction in more
detail. 

% There is just one small technical
% complication. If we encoded by the short code single symbols, the
% condition about the minimal distance would be violated on the
% short intervals at the end of these encodings. Therefore we will
% encode pairs of consecutive symbols so that the encoding of a pair
% $v_i,v_{i+i}$ overlaps with the encoding of the next pair
% $v_{i+1},v_{i+2}$. 

Let $c:\F_{2^\ell}^n\to \F_{2^\ell}^{dn}$ be given by its generator
matrix. Let $a:\F_2^{\ell}\to\F_2^{d''\ell}$ be a cyclic linear
tree code with minimum distance $\delta''>0$.
%\footnote{The use of cyclic codes is not essential, but it 
%  simplifies the construction.} 
By Theorem~\ref{t1} (more precisely, by its stronger version shown in
its second proof) we
can pick constants $d''$ and $\delta''>0$ (independent of $\ell$), and
construct the generator matrix of such a tree code $a$ in polynomial time
in~$n$. Further we need a good linear  
block code
$f:\F_2^{d\ell}\to\F_2^{d^*\ell}$. So $d^*$ and its minimum relative
distance $\epsilon>0$ are further constants. (Explicit polynomial time
constructions of such block codes are well-known,
see~\cite{macwilliams-sloane,vanLint}.) We may assume 
w.l.o.g. that $d''=d^*$. 

The generator matrix of $c'$ is defined as follows. We
set $d':=2d''$ (because we want to use odd bits for the code words of
$a$ and even bits for the code words of $f$).  We need to define
$c'(e_i^{n'})$ for $i=1\dts n'$. Let $i$ be given and let $i=\ell k
+j$ where $0\leq k\leq n$ and $1\leq j\leq\ell$.

\ben
\item On the odd bits of $c'(e_i^{n'})$, we put the string $a(e_1^\ell)$,
  where $e_1^\ell=1\bar{0}^{\ell-1}$, so that it starts at the
  $d'(i-1)+1$-st bit (the first bit on which $c'(e_i)$ should be
  nonzero). If $k=n$ and $j>1$, we truncate the string appropriately.
The rest of the odd bits are zeros. 
\item On the even bits, if $k<n$, we put $c(e_k^n)$ encoded by $f$ 
shifted by $d'\ell$. If $k=n$, we put zeros everywhere. In plain words, we put
  the beginning of $c(e_k^n)$ encoded by $f$ on the next $d'$-block
  after the block where the first nonzero bit occurs.
\een
We will now estimate the minimum weight of segments of the code words of
$c'$. Let $0\leq i<j\leq n'$ be given and suppose that $v$ is an input 
word in which the first nonzero element is on the coordinate
$i+1$. 

If $j\leq i+\ell$, then there are at least $\delta''d'(j-i)$ nonzero
elements among the odd elements in the interval in $c'(v)$
corresponding to the interval $(i,j]$ in $v$ because the vector
restricted to odd coordinates in this interval is a code word of the
tree code $a$.  Hence the relative weight is at
least $\delta''/2$, and if $j\leq i+2\ell$, the relative weight is at
least $\delta''/4$.

Now suppose that $j>i+2\ell$. Then there is at least one entire
$\ell$-block between $i$ and~$j$. Suppose there are $k$ such
$\ell$-blocks. They correspond to $k$ consecutive elements of a code
word of $c$ in which the first element is nonzero. Hence there are at
least $\delta k$ nonzero elements among them. Using the code $f$, they
are encoded in $c'(v)$ to a string with at least $\epsilon\delta d'\ell k$
nonzero elements. Since entire blocks cover at least $1/3$ of
the interval, this ensures positive relative minimum distance at least
$\epsilon\delta/3$ on even bits, hence $\epsilon\delta/6$ on all bits.

\end{proof}

It is possible that an explicit construction of tree codes is found
where the fields have polynomial size, but their characteristic
increases with $n$. E.g., the fields could be prime fields with the
prime $p$ larger than $n$. Then the above construction cannot be used
to construct a linear tree code over a constant size field, but it is
not difficult to modify it to produce a nonlinear tree code.

% But when talking about polynomial time constructions of general
% tree codes, we have to be more specific about what this means. We
% will say that \emph{a family of tree codes is constructible in
%   polynomial time} if there is a polynomial time algorithm that for
% every tree code  $c:\Pi^n\to\Sigma^n$ in the family and every given
% input word $x\in\Pi^n$ computes $c(x)$. We assume that an encoding of
% the alphabets $\Pi$ and $\Sigma$ by binary strings is given.

% \begin{proposition}\label{prop-3.2}
%   Let $b,d$ and $\delta>0$ be constants, then there exist constants
%   $d'$ and $\delta'>0$ such that the following is true. Suppose a
%   family of polynomial time constructible tree codes $\cal C$ is given
%   such that for every code $c:\Pi^n\to\Sigma^n$ in the family,
%   $q=|\Pi|\leq n^b$, $\Sigma=\Pi^d$ and the minimum relative distance of
%   $c$ is $\delta$. Then there exists a family $\cal C'$ of polynomial
%   time constructible tree codes such that for every tree code
%   $c:\Pi^n\to\Sigma^n$, $|\Pi|\leq n^b$ in $\cal C$, there is a
%   tree code $c':\{0,1\}^{n'}\to\{0,1\}^{d'n'}$ where $n'=\lfloor\log_2
%   q\rfloor (n+1)$ and the minimum relative distance of $c'$ is
%   $\delta'$.
% \end{proposition}
% \prfs
% Choose a one-to-one mapping from $\{0,1\}^{\log_2q}$ into
% $\Sigma$. Then proceed in the same way as in the proof above with the
% only difference that now we have to define every code word instead of
% the generators .
% \qed

\section{A characterization of the minimum distance}

We will characterize the minimum distance the tree codes defined by
parity-check matrices in normal forms. For the sake of simplicity, we will
assume that the rate of the tree codes is $1/2$ (i.e., $c:F^n\to F^{2n}$).  We
will use the following standard notation. Given a matrix $M$ and
indices of rows $i_1<\dots <i_\ell$ and columns $j_1<\dots <j_k$,
\[
M[i_1\dts i_\ell|j_1\dts j_k]
\]
denotes the submatrix of $M$ determined by these rows and columns.

\begin{proposition}\label{l-4.1}
Let $M$ be an $n\times 2n$ parity-check matrix of a linear tree code $c$
in a normal form. Then $\tilde{\delta}(c)$ is the
least $\delta>0$ such that there are $0\leq k<\ell\leq n$ and $t$
indices 
$2k<j_1<\dots <j_t\leq 2\ell$, $j_1\leq 2k+2$ such that 
\ben
%\item $t\leq \delta(2k-2\ell)$,
\item $t\leq 2\delta(\ell-k)$, and
\item in $M[k+1,k+2\dts \ell\ |\ j_1\dts j_t]$ the first column is a
  linear combination of the other columns.
\een
\end{proposition}

\begin{proof}
  Let $v$ be a nonzero code word of the code for which the minimum
  distance is attained. Let $j_1$ be the first coordinate of $v$ that
  is nonzero and let $2k<j_1\leq 2k+2$, $k<\ell\leq n$ and
  $j_1\dts j_t$ be all the nonzero coordinates of $v$ between
  $2k+1$ and $2\ell$ such that
\[
\tilde{\delta}(c)=\frac t{2(\ell-k)}.
\]
Since $M$ is a parity-check matrix of the code, the sum of
columns of $M$ with weights $v_t$ must be a zero vector. Note the
following two facts. First, the columns $2k+1,2k+2\dts 2n$ have
zeros on the rows $1\dts k$. Second, the columns $2\ell+1,2\ell+2\dts 2n$
have zeros on the rows $1\dts \ell$. From this, we get condition 2. Hence
$\tilde{\delta}(c)$ is at least the minimum $\delta$ that satisfies
the conditions of the lemma.

To show that it is at most $\delta$, suppose that $2k<j_1<\dots
<j_t\leq 2\ell$, $j_1\leq 2k+2$ are such that the two conditions are
satisfied. Let $\alpha_1,\alpha_2\dts \alpha_t$, $\alpha_1\neq 0$ be
the weights of a linear combination that makes the zero vector from
the columns of $M[k+1,k+2\dts \ell\,|\,j_1\dts j_t]$.
%  $j_1\dts j_t$
% restricted to the rows $\ell+1,i+2\dts k$.

We will show that there is a
code word $v$ that has zeros before the coordinate $j_1$, it has
nonzero 
on it, and all nonzeros between $2k+1$ and $2\ell$ are on coordinates
$j_1\dts j_t$. 
%It is clear what are the coordinates of $v$ up to$2\ell$. 
For $i=1\dts 2\ell$, we define $v_i=\alpha_{j_p}$ if $i=j_p$ for some
$1\leq p\leq t$, and $v_i=0$ otherwise.
This guarantees that the vector $Mv^\top$ has zeros on all
coordinates $1\dts \ell$, no matter how we define $v$ on the remaining
coordinates $2\ell+1\dts 2n$. Now we observe that the matrix  
$M[\ell+1\dts n\, |\, 2\ell+1\dts 2n]$ has full
rank, so a suitable choice of the coordinates $2\ell+1\dts 2n$ will make
the product $Mv^\top$ zero vector. Hence
\[
\tilde{\delta}(c)\leq\frac t{2(\ell-k)}\leq \delta.
\]
\end{proof}

\section{MDS tree codes}\label{s-MDS}

In this section we define tree codes that correspond to MDS block
codes and prove two characterizations of them.  Again, for the sake of
simplicity, we define it only for rate $1/2$ codes. First we prove a
general upper bound on the relative distance of linear tree codes of
rate $1/2$ that corresponds to the Singleton bound for block
codes. (As in the theory of block codes, this bound holds true also
for nonlinear tree codes and similar bounds can be proven for other rates.)

\bpr\label{p-5.0} 
For every linear tree code $c:F^n\to F^{2n}$,
$\tilde{\delta}(c)\leq\frac{n+1}{2n}$.  
\epr 
\prf
Let $M=(m_{ij})_{i,j}$ be a parity check matrix in a normal form. If
$m_{11}=0$ or $m_{12}=0$ we can construct a code word whose second,
respectively, first coordinate is zero and the other one is nonzero. 
(Suppose, e.g., that $m_{11}=0$. Define $v_1=1$ and $v_2=0$.  Now we
can extend $(v_1,v_2)$ to a code word because the matrix $M[2\dts
n\,|\,3\dts 2n]$ has full rank; see Proposition~\ref{p-2.3}.)
Hence $\tilde{\delta}(c)\leq
\frac 12\leq \frac{n+1}{2n}$. 

So suppose that $m_{11}\neq 0$ and $m_{12}\neq
0$. Let $3\leq j_2\leq 4$ \dots $2n-1\leq j_n\leq 2n$ be indices of
columns such that $M_{t,j_t}\neq 0$. We know that such columns exist
by Proposition~\ref{p-2.3}. Then the first column of $M$ is a linear
combination of columns $2,j_2,j_3\dts j_n$. Hence there is a code word
whose first coordinate is nonzero and it has at most $n+1$ nonzero
coordinates. Thus $\tilde{\delta}(c)\leq \frac{n+1}{2n}$.
\qed

The tree codes that meet the bound of Proposition~\ref{p-5.0}
naturally correspond to MDS block codes and therefore we make the following
definition.

\begin{definition}
A linear tree code $c:F^n\to
F^{2n}$ will be called an \emph{MDS tree code} if  $\tilde{\delta}(c)=
\frac{n+1}{2n}$.
\end{definition}
By Proposition~\ref{p-5.0}, the condition $\tilde{\delta}(c)=
\frac{n+1}{2n}$ is equivalent to $\tilde{\delta}(c)> \frac 12$.

\begin{proposition}\label{p-5.1}
  Let $M$ be a parity-check matrix of a linear tree code $c:F^n\to
  F^{2n}$ and let $M$ be in a normal form. Then $c$ is an MDS tree code
  if and only if for every $n$-tuple $1\leq j_1<\dots<j_n\leq 2n$
  satisfying 
\bel{e-prop5.1}
j_1\leq 2, j_2\leq 4\dts j_n\leq 2n, 
\ee
the columns $j_1\dts j_n$ are linearly independent.
\end{proposition}

\begin{proof}

First we show that the condition in the proposition implies the
following formally stronger condition:
\bi
\item[($\xi$)] for every $0\leq\ell<\ell+t\leq n$ and $2\ell< j_1<\dots<j_t$,
where $j_1\leq 2(\ell+1)\dts j_t\leq 2(\ell+t)$ the matrix
$
M[\ell+1\dts \ell+t \I j_1\dts j_t]
$
is nonsingular.
\ei
Indeed, given $j_1<...<j_t$ satisfying the general condition, we can
add $\ell$ elements before $2\ell$ and $n-\ell-t$ elements after $2(\ell+t)$ so that
the resulting $n$-tuple satisfies the condition of the
proposition. Let $N$ be the matrix consisting of these $n$ columns of
$M$. The matrix $N$ has the following block structure
consisting of square matrices
\[
\begin{pmatrix}
T_1 & 0 & 0  \\
A & M^* & 0  \\
B & C & T_2
\end{pmatrix}
\]
where $M^*=M[\ell+1\dts \ell+t \I j_1\dts j_t]$. Since $N$ is
nonsingular, $M^*$ must also be nonsingular.

Now suppose $M$ satisfies condition ($\xi$). Arguing by contradiction,
suppose that $\tilde{\delta}(c)\leq \frac 12$. By
Proposition~\ref{l-4.1}, we have $0\leq \ell<k\leq n$ and $t$ indices
$2\ell<j_1<\dots <j_t\leq 2k$, $j_1\leq 2\ell+2$, $t\leq 2\tilde{\delta}(c)
(k-\ell)\leq k-\ell$ such
that in $M[\ell+1,\ell+2\dts k\ |\ j_1\dts j_t]$ the first column is a
linear combination of the other columns.  
%(We are using the fact that
%$\lfloor\frac 12(2k-j_1+1)\rfloor=k-\ell$.)  
Let $s$ be the maximal
element $1\leq s\leq t$ such that for all $1\leq r\leq s$, the
inequality $j_r\leq 2(\ell +r)$ is true.  If $s<t$, then $j_{s+1}>
2(\ell+s+1)$. Hence any column $j_r$, for $r>s$, has zeros in rows
$\ell+1\dts \ell+s$. Thus the fact that the first column of
$M[\ell+1,\ell+2\dts k\ |\ j_1\dts j_t]$ is a linear combination of
the others implies that the same holds true for $M[\ell+1,\ell+2\dts
\ell+s\ |\ j_1\dts j_s]$. But this is impossible, because this matrix
is nonsingular according to ($\xi$).

To prove the opposite implication, suppose that we have an $n$-tuple
$1\leq j_1<\dots<j_n\leq 2n$ satisfying $j_1\leq 2, j_2\leq 4\dts
j_n\leq 2n$ such that the columns $j_1\dts j_n$ are linearly
dependent. Suppose that for some $\ell$, the column $j_{\ell+1}$ is a
linear combination of columns $j_{\ell+2}\dts j_n$. Then in 
 $M[\ell+1,\ell+2\dts n\ |\ j_{\ell+1}\dts j_n]$ the first column is a
  linear combination of the other columns, which violates the conditions
  of Proposition~\ref{l-4.1}.
\end{proof}

If
$M=(m_{i,j})_{i,j}$ is a parity check matrix in a normal form of a
code of rate $1/2$, we have either $m_{i,2i-1}\neq 0$ or $m_{i,2i}\neq
0$ for every $1\leq i\leq n$. Since permuting columns $2i-1$ and $2i$
does not change the relative distance, we can always assume w.l.o.g. that
\ben
\item[($\eta$)] all entries $m_{i,2i}$, $i=1\dts n$, are nonzero.  
\een 

Let $M$ be in a normal form and suppose that it satisfies
($\eta$). Using row operations we can eliminate all nonzero entries in
even columns, except for $m_{i,2i}$. Then we can multiply the rows to
get $m_{i,2i}=1$. The resulting matrix consists of a lower triangular
matrix interleaved with the identity matrix $I_n$. We will
characterize these triangular matrices of MDS tree codes.

A matrix $M$ is called \emph{totally nonsingular} if every square
submatrix of $M$ is nonsingular. A triangular matrix of dimension
$n\geq 2$ cannot be totally nonsingular because in a totally
nonsingular matrix every element is nonzero. However, there is a
natural modification that does make sense for triangular matrices. 
\begin{definition}
  An $n\times n$ lower triangular matrix $L$ is called
  \emph{triangular totally nonsingular} if for every $1\leq s\leq n$
  and every $1\leq i_1<\dots<i_s\leq n$, $1\leq j_1<\dots j_s\leq n$
  such that $j_1\leq i_1\dts j_s\leq i_s$, the submatrix $L[i_1\dts
  i_s |\, j_1\dts j_s]$ is nonsingular.
\end{definition}
Roughly speaking, $L$ is triangular totally nonsingular if it is
triangular and every square submatrix of $L$ that can be nonsingular,
is nonsingular. Upper triangular totally nonsingular matrices are
defined by reversing the inequalities between the indices $i$ and $j$,
i.e., requiring $j_1\geq i_1\dts j_s\geq i_s$.

% Recall from Section~\ref{sect2} that in the special case of rate $1/2$
% the parity check matrix can be presented as a unit matrix interleaved
% with a lower triangular matrix, see the condition $4^\circ$. We will
% call the lower triangular matrix the \emph{characteristic matrix of
%   the code.}

\begin{theorem}\label{t-5.1}
  Suppose that a parity check matrix of linear tree code $c:F^n\to
  F^{2n}$ has the form of a lower triangular matrix $T$ interleaved with
  the identity matrix $I_n$. Then $c$ is an MDS tree code if and
  only if $T$ is triangular totally nonsingular.
\end{theorem}
Let us note that a similar fact for MDS codes is well-known
(namely, the statement with totally nonsingular matrices instead of
triangular totally nonsingular matrices).%
\footnote{See Ch.11, \S 4, Theorem~8
  in~\cite{macwilliams-sloane}.}

\begin{proof}
  Let $j_1<\dots <j_p$ be some columns of $T$ and $k_1<\dots<k_q$ some
  columns of $I_n$ where $p+q=n$. Consider the determinant of the
  matrix formed by these columns. Observe that each nonzero term in
  the formula for this determinant must choose elements with
  coordinates $(k_1,k_1)\dts (k_q,k_q)$ from $I_n$ because these are
  the only nonzero elements in these columns. This implies that the
  determinant is equal, up to the sign, to
\[
\det(T[i_1\dts i_p\I j_1\dts j_p]),
\]
where $\{i_1\dts i_p\}=[1,n]\setminus\{k_1\dts k_q\}$,
$i_1<\dots<i_p$. (Note that we are now indexing columns by numbers
from $1$ to $n$ in both matrices $T$ and $I_n$.)  Hence to prove the
theorem it suffices to show that the condition \eqref{e-prop5.1} of
Proposition~\ref{p-5.1} on the indices of columns that should be
independent is equivalent to the condition on the indices of rows and
columns of submatrices that should be nonsingular in a triangular
nonsingular matrix.

First we note that the condition \eqref{e-prop5.1} translates to the following
\bel{bb}
\mbox{for all }s\leq p+q,\quad
|\{j_1\dts j_p\}\cap[1,s]|+|\{k_1\dts k_q\}\cap[1,s]|\geq s.
\ee
Since
\[
|\{k_1\dts k_q\}\cap[1,s]|=s-|\{i_1\dts i_p\}\cap[1,s]|,
\]
condition \eqref{bb} is equivalent to 
\bel{e-aa}
\mbox{for all }s\leq p+q,\quad
|\{j_1\dts j_p\}\cap[1,s]|\geq|\{i_1\dts i_p\}\cap[1,s]|. 
\ee
This inequality is satisfied for all $s\leq p+q$ if and only if it is
satisfied for all $s=i_1\dts i_p$. But for $s=i_r$ the inequality
\eqref{e-aa} is equivalent to the simple condition that $j_r\leq
i_r$, which is the condition required in the definition of triangular
totally nonsingular matrices.
\end{proof}

We note that we get a similar characterization of \emph{generator} matrices
of MDS linear tree codes of rate $1/2$.

\begin{corollary}
  Suppose that a linear tree code $c:F^n\to F^{2n}$ has a parity check
  matrix satisfying condition ($\eta$). Then $c$ is an MDS tree code
  if and only if it has a generator 
  matrix $N$ whose form is an upper triangular totally
  nonsingular matrix $S$ interleaved with $-I_n$ (minus the identity
  matrix).
\end{corollary}
\begin{proof}
  Let $T$ be a lower triangular totally nonsingular matrix. Let $T$
  interleaved with $I_n$ be a parity check matrix of an MDS tree code
  $c:F^n\to F^{2n}$. Then $N$ constructed from $(T^{-1})^T$ and $-I_n$
  is, clearly, a generator matrix that generates the code words of
  $c$. By Jacobi's equality (see, e.g., \cite{fallat-johnson}),
  $T^{-1}$ is upper triangular totally nonsingular.

The proof of the opposite direction is essentially the same and we
leave it to the reader.
\end{proof}

\section{Triangular totally nonsingular matrices}

By Proposition~\ref{prop-3.2} and Theorem~\ref{t-5.1}, the problem of
constructing an asymptotically good tree code reduces to the problem
of constructing a triangular totally nonsingular matrix over a field of
polynomial size. We are not able to construct such matrices and, in
fact, we are even not able to prove that they exist.

\begin{problem}
Do there exist triangular totally nonsingular matrices over fields
with polynomially many elements? If they do, construct them explicitly.
\end{problem}
According to Theorem~\ref{t-5.1}, the problem is equivalent to the
question whether there exist linear MDS tree codes over fields with
polynomially many elements. We believe that in order to prove that
such matrices (and such codes) exist, one has to define them
explicitly. In this section we will discuss some approaches to this
problem.

First we observe that triangular totally nonsingular matrices exist in
fields of every characteristic. A simple way of proving this fact is
to take a lower triangular matrix whose entries on and below the main
diagonal are algebraically independent over a field of a given
characteristic. Below is a slightly more explicit example.

\begin{lemma}
Let $x$ be an indeterminate and let $W_n(x)=(w_{ij})_{i,j=1}^n$ be the
lower triangular $n\times n$ matrix defined by
\[
w_{ij}=x^{(n-i+j-1)^2}
\]
for $i\geq j$, and $w_{ij}=0$ otherwise. Then for every $1\leq
i_1<\dots<i_s\leq n$, $1\leq j_1<\dots j_s\leq n$ such that $j_1\leq
i_1\dts j_s\leq i_s$, the determinant
\[
\det(W_n(x)[i_1\dts i_s |\, j_1\dts j_s])
\]
is a nonzero polynomial in every characteristic.
\end{lemma}
\begin{proof}
  We will show that the monomial of the highest degree occurs exactly
  once in the formula defining the determinant. We will use
  induction on $s$. For $s=1$, it is trivial. Suppose that $s>1$ and let
  $1\leq i_1<\dots<i_s\leq n$, $1\leq j_1<\dots j_s\leq n$ such that
  $j_1\leq i_1\dts j_s\leq i_s$ be given. Denote by $M:=W_n(x)[i_1\dts
  i_s |\, j_1\dts j_s]$.

  First we need to prove an auxiliary fact.  We will call $(i,j)$ an
  \emph{extremal position} in the matrix $M$ if 
\ben
\item $i\in\{i_1\dts i_s\}$, $j\in\{j_1\dts j_s\}$, $w_{ij}\neq 0$
  and,
\item for every $(i',j')$, if $i'\in\{i_1\dts i_s\}$, $j'\in\{j_1\dts
  j_s\}$ $i'\leq i$, $j'\geq j$ and $(i,j)\neq (i',j')$, then
  $w_{i'j'}=0$.
\een
Since $M$ has nonzero elements, it must have at least one
extremal position.
We will show that the monomial with the highest degree
  must contain all $w_{ij}$ where $(i,j)$ is extremal.

Suppose that $(i,j)$ is an extremal position and
  some nonzero monomial $m$ does not contain  $w_{ij}$. Then $m$ must contain
some  elements $w_{ij'}$ and $w_{i'j}$ where $j'<j$ and
$i'>i$. Observe that
\[
w_{ij'}w_{i'j}=x^{(n-i+j'-1)^2+(n-i'+j-1)^2}\quad\mbox{ and }\quad
w_{ij}w_{i'j'}=x^{(n-i+j-1)^2+(n-i'+j'-1)^2}.
\]
The difference between the second and the first exponent is
\[
-2ij-2i'j'+2ij'+2i'j=2(i'-i)(j-j')>0.
\]
Hence we get a monomial of higher degree if we replace
$w_{ij'}w_{i'j}$ by $w_{ij}w_{i'j'}$. This establishes the fact.

Now we can finish the proof. If every row index   $i\in\{i_1\dts
i_s\}$ and every column index $j\in\{j_1\dts j_s\}$ occurs in an
extremal position, then $M$ is lower triangular and the determinant has
only one monomial, the product of all elements in extremal positions. 
Otherwise, delete from $M$ all rows and columns
whose indices occur in the extremal position. The remaining matrix is
nonempty and, by the induction assumption, has
a unique monomial $m'$ of the maximal degree. The unique monomial of
the highest degree of $M$ is obtained from $m'$ by multiplying it by
$w_{ij}$ for all extremal positions $(i,j)$.
\end{proof}

Given a prime $p$ and a number $n\geq 1$, we can take an irreducible
polynomial $f(x)$ over $\F_p$ of degree higher than the degrees of the
determinants of the
square submatrices of $W_n(x)$. Then the matrix is triangular totally
nonsingular over the field $\F_p[x]/(f(x))$. The size of this field is
exponential, because the degree of $f(x)$ is polynomial in $n$.

\medskip
While we do not know the answer to the problem above, 
constructions of totally nonsingular matrices over fields of
linear size are known. Here are some examples.

Let $F$ be an arbitrary field and let $a_1\dts a_m,b_1\dts b_n\in F$ be
such that $a_i\neq b_j$ for all $i,j$. The matrix $\left(\frac
  1{a_i-b_j}\right)_{i,j}$ is called a \emph{Cauchy matrix} (see,
e.g., \cite{macwilliams-sloane}, page 323). If all
$a_i$ and all $b_j$ are distinct elements, then the matrix is
nonsingular. Since every submatrix of a Cauchy matrix is a Cauchy
matrix, the condition also implies that the Cauchy matrix with
distinct elements $a_1\dts a_m,b_1\dts b_n$ is totally nonsingular.
Thus given a field with at least $2n$ elements, we are able to
construct a totally nonsingular matrix of the dimension $n$. A special
case of a Cauchy matrix is the \emph{Hilbert matrix} $\left(\frac
  1{i+j-1}\right)_{i,j}$. 

More generally, we call a matrix of the form  $\left(\frac
  {g_ih_j}{a_i-b_j}\right)_{i,j}$, $g_i,h_j\in F$ a \emph{Cauchy-like
matrix}. Such a matrix is totally nonsingular if and only if all the
elements $a_1\dts a_m,b_1\dts b_n$ are distinct and the elements
$g_1\dts g_m,h_1\dts h_n$ are nonzero. One special case is the
\emph{Singleton matrices} which are matrices of the form  $\left(\frac
  1{1-a^{i+j-2}}\right)_{i,j}$ where the order of the element $a$ is
larger than $2n-2$, see~\cite{singleton}.

One can also construct a totally nonsingular matrix from a
parity-check matrix of an MDS block code.  The standard construction
of MDS block codes is the Reed-Solomon codes whose parity-check and
generator matrices are Vandermonde matrices. We will explain the
construction of a totally nonsingular matrix from an MDS code on this
special case.  We denote by $V_m(x_1\dts x_n)$ the \emph{Vandermonde
  matrix} $\left(x^i_j\right)_{i=0\dts m-1;j=1\dts n}$. Consider the
Vandermonde matrix $V_m(a_1\dts a_m,b_1\dts b_n)$ where all the
elements $a_1\dts a_m,b_1\dts b_n$ are distinct. If we diagonalize the
first $m$ columns using row operations, then the submatrix consisting
of the last $n$ columns becomes a totally nonsingular matrix. (This
follows from the fact that any set of $m$ columns is independent using
the argument used in the proof of Theorem~\ref{t-5.1}.)  The
diagonalization can be represented as multiplying $V_m(a_1\dts
a_m,b_1\dts b_n)$ by $V_m(a_1\dts a_m)^{-1}$ from the left. So this
means that the matrix
\[
V_m(a_1\dts a_m)^{-1}V_m(b_1\dts b_n)
\]
is totally nonsingular. One can check that this matrix is also
Cauchy-like (see, e.g., \cite{fiedler}, page 159, or \cite{roth-seroussi}).

%note that the notation used there is different.)

The theory of \emph{totally positive matrices} is another source of totally
nonsingular matrices. A matrix over $\R$ is totally positive if all its square
submatrices have a positive determinant. An example of a totally
positive matrix is the \emph{Pascal matrix} 
% $P^{m,n}:=\left({i+j\choose j}\right)_{i=1\dts m;j=1\dts n}$. 
$P_{n}:=\left({i+j\choose j}\right)_{i,j=0}^n$.  The elements of
this matrix are integers, so we can consider it over any field. Hence
if $p$ is a sufficiently large prime, then $P_n$ is totally
nonsingular in $\F_p$, but we do not know if $p$ can be
subexponential. The same problem arises for other examples of totally
positive matrices with rational coefficients, and triangular totally
positive matrices as well---we do not know if we can
find a polynomially large prime for which these matrices are still
totally nonsingular. 

Another important example of a totally positive matrix is the
Vandermonde matrix $V_n(a_1\dts a_n)$ with $0<a_1<\dots<a_n$. In 
general, a real Vandermonde matrix is not always totally nonsingular.

\emph{Triangular totally positive matrices} are
defined in a similar way as triangular totally nonsingular matrices.
The most interesting fact for us is that total positivity is preserved
by $LU$ decompositions.  We
quote the following result of Cryer~\cite{cryer}, see
also~\cite{fallat-johnson}, Corollary~2.4.2.

\begin{theorem}\label{t-4.3}
  A matrix $M$ is totally positive if and only if it has an $LU$
  factorization in which the terms $L$ and $U$ are both triangular
  totally positive. 
\end{theorem}
Recall that the terms in an $LU$ decomposition are unique up to
multiplication by diagonal nonsingular matrices. Hence if $M$ is
totally positive, then $L$ and $U$ are both triangular totally
nonsingular for any $LU$ factorization. % (The factors that are 
% totally positive are those whose all entries are nonnegative.)

This theorem suggests the possibility of constructing triangular totally
nonsingular matrices from totally nonsingular matrices using an $LU$
decomposition. Unfortunately, the statement corresponding to the
theorem above is not true for totally nonsingular matrices.

\begin{fact}
There exists a $4\times 4$ matrix over rational numbers which is
totally nonsingular, but its $L$-factor is not triangular totally
nonsingular. 
\end{fact}
\begin{proof}
Suppose that we use Gaussian elimination to transform a totally nonsingular
$4\times 4$ matrix to its $L$-factor and we arrive at the
following situation:

\[
\begin{pmatrix}
1 & 0 & 0 & 0 \\
1 & 2 & 2 & * \\
2 & 3 & 5 & * \\
3 & 4 & 7 & *
\end{pmatrix}
\qquad\mapsto\qquad
\begin{pmatrix}
1 & 0 & 0 & 0 \\
1 & 2 & 0 & 0 \\
2 & 3 & 2 & * \\
3 & 4 & 3 & *
\end{pmatrix}
\qquad\mapsto\qquad
\begin{pmatrix}
1 & 0 & 0 & 0 \\
1 & 2 & 0 & 0 \\
2 & 3 & 2 & 0 \\
3 & 4 & 3 & *
\end{pmatrix}
\]
The entries denoted by $*$ will be determined later. The last matrix
is not triangular totally nonsingular, because it contains
$\begin{pmatrix}
  2 & 2 \\
  3 & 3
\end{pmatrix}$
as a submatrix. So we only need to show that the first matrix can be
obtained from a totally nonsingular matrix by elimination in the first
row. Equivalently, we need to show that there exists a totally
nonsingular matrix $M$ of the form
\[
\begin{pmatrix}
1 & a & b & c \\
1 & 2+a & 2+b & d \\
2 & 3+2a & 5+2b & e \\
3 & 4+3a & 7+3b & f
\end{pmatrix}
\]
Let $N$ be the $4\times 3$ submatrix made of the first 3 columns 
of the
first matrix in the chain above. We will use the fact that $N$ is
totally nonsingular, which is easy to verify. 

We first show that one can choose $a$ and $b$ so that the first three
columns of $M$ form a totally nonsingular matrix $M'$. Let $A$ be a submatrix
of $M'$. If $A$ contains the first column, then it is nonsingular,
because $N$ is. If $A$ is a $2\times 2$ submatrix in the second two
columns, we can transform it, using a row operation, into a matrix in
which $a$ and $b$ only appear in the first row and there are nonzero
elements in the second. Thus $A$ is nonsingular provided that a
certain nontrivial linear function in $a$ and $b$ does not
vanish. Hence there is a finite number of nontrivial linear equations such that
if we pick $a$ and $b$ so that none is satisfied, then $M'$ is totally
nonsingular. 

Let $a$ and $b$ be fixed so that $M'$ is totally nonsingular. Now
consider a submatrix $A$ of $M$ that contains the last column. If $A$
is singular, then a certain linear function in $c,d,e,f$ must
vanish. This function is nontrivial, because its coefficients are
subdeterminants of $M'$, possibly with negative signs, or it is just
one of the variables $c,d,e,f$. Thus, again,
there is a choice of $c,d,e$ and $f$ that makes all these functions
nonzero and hence all these matrices $A$ nonsingular.
\end{proof}

On the positive side, we can prove the following simple fact, which
is, however,
not sufficient for constructing good tree codes.

\begin{proposition}
Let $F$ be an arbitrary field and let $M$ be an $n\times n$ totally
nonsingular matrix over $F$. Let $M=LU$ be an $LU$-factorization. Then
for every $1\leq k\leq n$, and $1\leq j\leq i_1<\dots <i_k\leq n$, the
matrix
\[
L[i_1\dts i_k\, |\, j,j+1\dts j+k-1]
\]
is nonsingular.
\end{proposition}
\begin{proof}
An $L$ factor of $M$ can be obtained by Gaussian elimination using
column operations. Hence every matrix of the form 
$L[i_1\dts i_k\, |\, 1\dts k]$ is a matrix obtained from 
$M[i_1\dts i_k\, |\, 1\dts k]$ using column operations. Thus every
such matrix is nonsingular. 

Now consider a submatrix $L[i_1\dts i_k\, |\, j,j+1\dts j+k-1]$, where
$1\leq j\leq i_1<\dots i_k\leq n$. Extend this matrix to  
$L[1\dts j-1,i_1\dts i_k\, |\, 1,2\dts j+k-1]$, which is nonsingular
by the previous observation. Also the matrix 
$L[1\dts j-1\, |\, 1\dts j-1]$ is nonsingular. Since 
$L[1\dts j-1\, |\, j,j+1\dts j+k-1]$ is a zero matrix, this implies
that $L[i_1\dts i_k\, |\, j,j+1\dts j+k-1]$ is nonsingular.
\end{proof}

\section{Birkhoff interpolation}

As mentioned above, Reed-Solomon codes are the standard construction
of MDS block codes and they are based on Vandermonde matrices. To
prove their properties one uses Lagrange interpolation. A natural
question then is whether there are similar concepts connected with MDS
tree codes. In this section we will argue that the problem
corresponding to Lagrange interpolation is Birkhoff interpolation.

\emph{Birkhoff interpolation} is the following problem.
Given distinct complex
numbers $a_1\dts a_m$, integers $0\leq i_{1,0}<\dots <i_{1,j_1}$,
\dots,    
$0\leq i_{m,0}<\dots <i_{m,j_m}$, and arbitrary complex numbers $A_{10}\dts
A_{1j_1}\dts A_{m0}\dts A_{mj_m}$, find a polynomial $f(x)$ of degree
$m\leq n-1+j_1+\dots +j_m$ such that
\[
\begin{array}{ccc}
f^{(i_{1,0})}(a_1)=A_{10},&\dots&f^{(i_{1,j_1})}(a_1)=A_{1j_1},\\
\vdots& & \\
f^{(i_{m,0})}(a_m)=A_{m0},&\dots&f^{(i_{m,j_m})}(a_m)= A_{mj_m},
\end{array}
\]
where $f^{(i)}$ denotes the $i$th derivative of $f$.

The special case in which $i_{k,l}=l$ for all $k,l$, called
\emph{Hermite interpolation}, has always a unique solution. However in
general, special conditions for the numbers $i_{k,l}$ and $a_k$ must
be imposed if we want to have a solution for every choice of
$A_{10}\dts A_{1j_1}\dts A_{m0}\dts A_{mj_m}$. We are interested in
the special case of $m=2$, which was solved by P\'olya~\cite{polya}.

\begin{theorem}\label{t-interpol}
  Let $a,b\in\C$, let $0\leq j_1<\dots j_p$ and $0\leq
  k_1<\dots<k_q$ be integers and let $A_1\dts A_p,B_1\dts
  B_q\in\C$.  Suppose that $a\neq b$ and the integers satisfy the following
  condition (which we already used in the proof of Theorem~\ref{t-5.1})
\[
\mbox{for all }s\leq p+q,\quad
|\{j_1\dts j_p\}\cap[0,s]|+|\{k_1\dts k_q\}\cap[0,s]|\geq s+1.
\]
Then there exists a unique polynomial $f(x)$ of degree $n\leq
p+q-1$ such that $f^{(j_s)}(a)=A_s$ for all $s=1\dts p$ and 
 $f^{(k_t)}(b)=B_t$ for all $t=1\dts q$.

 % Furthermore, if {\rm (\dag)} is not satisfied, then for some values
 % of $A_s$ and $B_t$ the problem does not have any solution. - not
 % true!!! a solution exists if the condition is shifted; then it is not
 % unique 
\end{theorem}

This theorem is stated for the field of complex numbers, but it holds
true also for fields of characteristic $r>0$, in particular for prime
fields $\F_r$, if $r$ is sufficiently large. For those primes $r$ for
which it is true, one can construct an MDS tree code. Unfortunately
we only know that for primes exponentially big in $p+q$.

\medskip
To see the connection of MDS tree codes to Birkhoff interpolation,
consider the matrix of the linear equations that one needs to solve
in order to find the interpolating polynomial.
For $n\geq 1$, let  $M(x,y)$ be the $(n+1)\times 2(n+1)$ matrix with
entries $M_{i,2j}=(x^i)^{(j)}$ and $M_{i,2j+1}=(y^i)^{(j)}$ (the
derivatives of terms $x^i$ and $y^i$) where
$i,j=0\dts n$.%
\footnote{Here we are exceptionally numbering rows and columns
  starting with zero.} 
So our matrix is
\medskip
\[
M(x,y)=
\begin{pmatrix}
1 & 1 & 0 & 0 & 0 & 0 & \dots  &0 & 0\\
x & y & 1 & 1 & 0 & 0 & \dots  &0 & 0\\
x^2&y^2&2x &2y & 2 & 2 & \dots  &0 & 0\\
\hdotsfor{9}\\
x^n&y^n&nx^{n-1}&ny^{n-1}&n(n-1)x^{n-2}&n(n-1)y^{n-2}&\dots&n!&n!
\end{pmatrix}
\]
\bigskip
Let
$
f(x)=c_nx^n+\dots +\dots c_0
$
be a polynomial. Then 
\[
(c_0\dts c_n)M(x,y)=(f(x),f(y),f'(x),f'(y)\dts f^{(n)}(x),f^{(n)}(y)).
\]

Theorem~\ref{t-interpol} tells us for which submatrices of $M(a,b)$,
$a,b\in\C$, $a\neq b$, the interpolation problem has a solution, hence
which sets of columns of $M(a,b)$ are independent. Recall that in the
proof of Theorem~\ref{t-5.1} we observed that \eqref{bb} is
equivalent to \eqref{e-prop5.1} of that proposition. Hence we have:

\bpr
Let $F$ be a field, $a,b\in F$, $a\neq b$, and $n$ a positive
integer. Then P\'olya's interpolation theorem
(Theorem~\ref{t-interpol}) holds true for every $p$ and $q$ such that
$p+q\leq n$ if and only if $M(a,b)$ is a parity check matrix of an MDS
tree code.
\epr

In order to get an idea for which fields the theorem can be true, we
will sketch a proof of Theorem~\ref{t-interpol}. Since any pair of
distinct elements of $F$ can be mapped to any other pair, we can
w.l.o.g. assume that $a=1$ and $b=0$. Next we divide each
column $2j$ and $2j+1$ by $j!$.  (In other words, we are replacing
standard derivatives by Hasse derivatives.)

In the resulting matrix the even columns are the matrix
$L_n:=\left({i\choose j}\right)_{i,j}$, where the binomial
coefficients are defined to be zero for $j>i$, and the odd columns
form the identity matrix. One can easily check that $L_n$ is an
$L$-factor of an $LU$ factorization of the Pascal matrix $P_n$. (This
can be shown by applying the binomial formula to the equality
$(x+y)^{i+j}=(x+y)^i(x+y)^j$.) Since $P_n$ is totally positive, $L_n$
is triangular totally positive by Theorem~\ref{t-4.3}. This implies
that condition~\eqref{bb} suffices for the solubility of Birkhoff
interpolation. 

\bigskip
So the problem boils down to the question, for which primes $r$, the
matrix $L_n$ is triangular totally nonsingular over the field
$\F_r$. This seems to be a very difficult problem and we do not dare
to conjecture that $r$ may be of polynomial size. A more promising
approach is to study the cases of Birkhoff interpolation that are
solvable in fields of polynomial size and see if they suffice to
ensure a positive minimum distance of the corresponding tree codes.

% \medskip
% Using $M(a,b)$ (as defined above) as a parity check matrix looks like a
% natural way of defining an MDS tree code. Unfortunately, we are not
% able to show that it works over polynomially large fields. To prove
% that it does work, one would have to show that there is a polynomially
% large prime $p$ that does not divide the non-zero determinants of an $L$ factor
% of the Pascal matrix $P_n$. If true, this would probably require to solve a 
% difficult number-theoretical problem. It is more likely that there is
% a polynomially large prime that does not divide \emph{most of these}
% determinants. This might suffice to show that the tree code has an
% asymptotically positive relative minimum distance. Another possibility
% could be to find a generator matrix whose entries have polynomially large
% absolute values. If such a matrix is found one could apply the reduction
% (Proposition~\ref{prop-3.2}) without taking the vectors modulo a
% prime. Again, we do not know if such generator matrices exist for
% $M(a,b)$, or any MDS code obtained from triangular totally positive
% matrices. 

\end{document}